\documentclass[runningheads]{llncs}

\usepackage[english]{babel}

\usepackage{comment}

\usepackage{amsmath}

\usepackage{graphicx}
\usepackage[colorlinks=true, allcolors=blue]{hyperref}
\newtheorem{observation}{Observation}

\newenvironment{hidden}{}

\excludecomment{hidden}

\usepackage{tabularx}

\newcommand{\Trule}{\rule{0pt}{3ex}}
\newcommand{\Brule}{\rule[-1.5ex]{0pt}{0pt}}
\newcommand{\prob}[3]{
\begin{center}
\begin{tabularx}{\textwidth}{|l X|}
	\hline
	\multicolumn{2}{|l|}{#1\Trule} \\
	\textbf{Input:}&{#2}\\
	\textbf{Task: }&{#3\Brule}\\
	\hline
\end{tabularx}
\end{center}
}

\usepackage{xspace}

\usepackage{enumitem}
\setlist[enumerate]{nosep}

\newcommand{\AgreeForest}{\textsc{MRAF}\xspace}

\title{Relaxed Agreement Forests}
\author{Virginia Ardévol Martínez\inst{2}, Steven Chaplick\inst{1}, Steven Kelk\inst{1}, Ruben Meuwese\inst{1}, Mat\'u\v{s} Mihal\'ak\inst{1}, Georgios Stamoulis\inst{1}}

\institute{Dept. Advanced Computing Sciences, Maastricht University, the Netherlands,\\ 
\and Universit\'{e} Paris-Dauphine, PSL University, CNRS, LAMSADE, France.
}

\authorrunning{V. Ardévol Martínez et al.}

\begin{document}
\maketitle

\begin{abstract}
There are multiple factors
%, both experimental and biological, 
which can cause the phylogenetic inference process to produce two or more conflicting hypotheses of the evolutionary history of a set $X$ of biological entities. 
That is:  phylogenetic trees with the same set of leaf labels $X$ but with distinct topologies. 
This leads naturally to the goal of quantifying the difference between two such trees $T_1$ and $T_2$. 
Here we introduce the problem of computing a \emph{maximum relaxed agreement forest} (\AgreeForest) and use this as a proxy for the dissimilarity of $T_1$ and $T_2$, which in this article we assume to be unrooted binary phylogenetic trees. 
\AgreeForest asks for a partition of the leaf labels $X$ into a minimum number of blocks $S_1, S_2, \ldots S_k$ such that for each $i$, the subtrees induced in $T_1$ and $T_2$ by $S_i$ are isomorphic up to suppression of degree-2 nodes and taking the labels $X$ into account. %(i.e. are homeomorphic). 
%We denote this minimum $d_{MRAF}(T_1, T_2)$. 
Unlike the earlier introduced maximum agreement forest (MAF) model,
%we do not demand that
the subtrees induced by the $S_i$ %are mutually disjoint; in MRAF they 
are allowed to overlap. 
We prove that it is NP-hard to compute MRAF, by reducing from the problem of partitioning a permutation into a minimum number of monotonic subsequences (PIMS). 
%
%Beyond the hardness, 
Furthermore, we show that MRAF has a polynomial time $O( \log n)$-approximation algorithm where $n=|X|$ and permits exact algorithms with single-exponential running time.
When at least one of the two input trees is a caterpillar, we prove that testing whether a MRAF has size at most $k$ can be answered in polynomial time when $k$ is fixed. 
We also note that on two caterpillars the approximability of MRAF is related to that of PIMS. 
Finally, we establish a number of bounds on MRAF, compare its behaviour to MAF both in theory and in an experimental setting and discuss a number of open problems. 
\end{abstract}

\section{Introduction}
The central challenge of phylogenetics, which is the study of phylogenetic (evolutionary) trees, is to infer a tree whose leaves are bijectively labeled by a set of species $X$ and which accurately represents the evolutionary events that gave rise to $X$ \cite{steel2016phylogeny}.
%: the leaves correspond to contemporary species and the internal vertices of the tree represent hypothetical common ancestors \cite{steel2016phylogeny}.
There are many existing techniques to infer phylogenetic trees from biological data and under a range of different objective functions \cite{lemey2009phylogenetic}. 
The complexity of this problem arises from the fact that we typically only have indirect data available, such as DNA sequences of the species $X$. 
Different techniques regularly yield trees with differing topologies, or the same technique constructs different trees depending on which part of a genome the DNA data is extracted from \cite{nakhleh2013computational}. 
% , since reticulate evolutionary phenomena such as hybridization and horizontal gene transfer can cause different genes in the same genome to have different evolutionary histories 
Hence, it is insightful to formally quantify the dissimilarity between (pairs of) phylogenetic trees, stimulating research into various distance measures. 
%on pairs of phylogenetic trees.

%, such as the Tree Bisection and Reconnection (TBR) distance or the Subtree Prune and Regraft (SPR) distance \cite{john2017shape}. 

Here we propose a new dissimilarity measure between unrooted phylogenetic trees $T_1, T_2$ which is conceptually related to the well-studied \emph{agreement forest} abstraction. An agreement forest (AF) is a partition of $X$ into blocks which induce, in the two input trees, non-overlapping isomorphic subtrees, modulo edge subdivision and taking the labels $X$ into account;
computing such a forest of minimum size (a MAF) is NP-hard \cite{HeinJWZ96}. The AF abstraction originally derives its significance from the fact that, in unrooted (respectively, rooted) phylogenetic trees, an AF of minimum size models \emph{Tree Bisection and Reconnection} (TBR) (respectively, \emph{rooted Subtree Prune and Regraft}, rSPR) distance \cite{AllenSteel2001,bordewich2005computational}. 
For background on AFs we refer to recent articles such as \cite{chen2016approximating,bulteau2019parameterized}. Here we propose the \emph{relaxed agreement forest} abstraction (RAF). The only difference in the definition is that we no longer require the partition of $X$ to induce non-overlapping subtrees; they only have to be isomorphic (see Fig.~\ref{fig:diff}). We write MRAF to denote a relaxed agreement forest of minimum size. As we shall observe later, in the worst case
MRAF can be constant while MAF grows linearly in $|X|$.

The fact that RAFs are allowed to induce overlapping subtrees is potentially interesting from the perspective of biological modelling. Unlike an AF, multiple subtrees of the RAF can pass through a single branch of $T_1$ (or $T_2$). This allows us to view $T_1$ and $T_2$ as the union of several interleaved, overlapping, common evolutionary histories. It is beyond the scope of this article to expound upon this,
but it is compatible with the trend in the literature of phylogenetic trees (or networks) having multiple distinct %evolutionary 
histories woven within them \cite{degnan2009gene,nakhleh2013computational}.

Our results are as follows. First, we show that it is NP-hard to compute a MRAF. We reduce from the problem of partitioning a permutation into a minimum number of monotone subsequences (PIMS).
%Beyond the hardness,
We show that MRAF has a polynomial time $O( \log n)$-approximation algorithm where $n=|X|$ and permits exact algorithms with single-exponential running time. 
When at least one of the two input trees is a caterpillar, we prove that ``Is there a RAF with at most $k$ components?'' can be answered in polynomial time when $k$ is fixed, i.e., in XP parameterized by $k$. 
This is achieved by dynamic programming.
We also note that on two caterpillars the approximability of MRAF is closely related to that of  PIMS. Along the way we establish a number of bounds on MRAF, compare its behaviour to MAF and
%in the penultimate section
undertake an empirical analysis on two existing datasets. %Finally, we discuss a number of open problems.

\begin{figure}[t]
\centering
\includegraphics[scale=0.8]{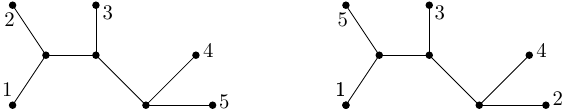}
\caption{
%Two phylogenetic trees on 5 leaves. A MRAF of these two trees has 2 components e.g., $\{1,5\}$ and $\{2,3,4\}$.  A MAF for these two trees has 3 components, however (e.g., $\{1\}$, $\{5\}$, $\{2,3,4\}$): the MRAF is smaller due to the subtrees induced by the components being allowed to overlap.
The two trees, while isomorphic, are not isomorphic when taking the leaf-labeling into account, and thus both MRAF and MAF cannot be of size one. A MRAF of these trees has 2 components, e.g., $\{1,2,3\}$ and $\{4,5\}$. A MAF of these trees has 3 components, e.g., $\{1,2,3\}$, $\{4\}$, and $\{5\}$.
}
\label{fig:diff} 
\end{figure}

\begin{hidden}
\begin{figure}[h]
\centering
\includegraphics[scale=0.4]{CrappyPicture.pdf}
\caption{Two phylogenetic trees on 14 leaves (phyB\_waxy from the grass dataset analysed later in the article). A MRAF of these two trees has 2 components, shown in red and purple. A MAF for these two trees has 3 components, however: the MRAF is smaller due to images of the components being allowed to overlap.}
\label{fig:intro} 
\end{figure}
\end{hidden}

\section{Preliminaries, basic properties and bounds}
Let $X$ be a set of labels (\emph{taxa}) representing species. 
An \textit{unrooted binary phylogenetic tree} $T$ on $X$ is a simple, connected, and undirected tree whose leaves are bijectively labeled with $X$ and whose other vertices all have degree 3. 
When it is clear from the context we will simply write (\emph{phylogenetic}) \emph{tree} for shorthand. For two trees $T$ and $T'$ both on the same set of taxa $X$, we write $T = T'$ if there is an isomorphism between $T$ and $T'$ that preserves the labels $X$. 
Tree $T$ is a \emph{caterpillar} if deleting the leaves of $T$ yields a path. We say that two distinct taxa $\{a, b\} \subseteq X$ form a \emph{cherry} of a tree $T$ if they have a common parent. 
The \emph{identity caterpillar} on $n$ leaves is simply the caterpillar with leaves $1, \ldots n$ in ascending order with the exception of the two cherries $\{1,2\}$ and $\{n-1, n\}$ at its ends; see e.g. the tree on the left in Fig.~\ref{fig:diff}.
Note that caterpillars are almost total orders, but not quite: the leaves in the cherries at the ends are incomparable. 
Managing this subtle difference is a key aspect of our results.

A {\it quartet} is an unrooted binary phylogenetic tree with exactly four leaves. Let $T$ be a
%unrooted binary
phylogenetic tree on $X$. If $\{a,b,c,d\}\subseteq X$ are four distinct leaves, we say that quartet $ab|cd$ \emph{is induced by (or simply `is a quartet of') $T$} if in $T$ the path from $a$ to $b$ does not intersect the path from $c$ to $d$.
%The definition of $ac|bd$ and $ad|bc$ is analogous.
Note that, for any four distinct leaves $a,b,c,d \in X$,
%each of the $\binom{n}{4}$ size-4 subsets of $X$, 
exactly one of the three quartets $ab|cd, ac|bd, ad|bc$ will be a quartet of $T$. It is well-known
%in the literature
that $T_1 = T_2$ if and only if both trees induce exactly the same set of quartet topologies \cite{steel2016phylogeny}.
For example, in Fig.~\ref{fig:diff} $12|45$ is a quartet of the first tree but not a quartet of the second tree. For $X' \subseteq X$, we write $T[X']$ to denote the unique, minimal subtree of $T$ that connects all elements in the subset $X'$. We use $T|X'$ to denote the phylogenetic tree on $X'$ obtained from $T[X']$ by suppressing %non-root 
degree-2 vertices. If $T_1|X' = T_2|X'$ then we say that the subtrees of $T_1, T_2$ induced by $X'$ are \emph{homeomorphic}.  

Let $T_1$ and $T_2$ be two phylogenetic trees on $X$. Let $\mathcal{F}= \{S_1, \ldots, S_k\}$ be a partition of $X$, where each block $S_i$, %with $i \in \{1, 2, \ldots, k\}$,
is referred to as a \emph{component} of $\mathcal{F}$. We say that  $\mathcal{F}$ is an \textit{agreement forest} (AF) for $T_1$ and $T_2$ if these conditions hold:
\begin{enumerate}
    \item For each $i \in \{1, 2, ..., k\}$ we have $T_1|S_i = T_2|S_i$.
    \item For each pair  $i,j \in \{1, 2, ..., k\}$ with $i \neq j$, we have that $T_1[S_i]$ and $T_1[S_j]$ are vertex-disjoint in $T$, and $T_2[S_i]$ and $T_2[S_j ]$ are vertex-disjoint in $T_2$.
\end{enumerate}
The \textit{size} of $\mathcal{F}$ is simply its number of components, i.e., $k$. Moreover, an AF with the minimum number of components (over all AFs for $T_1$ and $T_2$) is called a \textit{maximum
%\footnote{\emph{Maximum} is traditionally used to emphasize `maximization of agreement'.}
agreement forest} (MAF) for $T_1$ and $T_2$. For ease of reading we will also write MAF to denote the size of a MAF. This is NP-hard to compute \cite{HeinJWZ96,AllenSteel2001}.
%The number of components of a maximum agreement forest for $T_1$ and $T_2$ is usually denoted by $d_{MAF}(T_1, T_2)$, although we will often simply write MAF when this is clear from the context. 

A \emph{relaxed agreement forest} (RAF) is defined similarly to an AF, except without condition 2.
%(which enforces that the induced components are non-overlapping in each tree)
A RAF with a minimum number of components is a \emph{maximum relaxed agreement forest} (MRAF). We also use MRAF for the size of a MRAF.
%The number of components in a MRAF is denoted $d_{MRAF}(T_1, T_2)$ although we will simply write MRAF where this is understood. Expressed as a decision problem:
%As usual, we can also view MRAF as an optimization problem:

\prob{\textsc{Maximum Relaxed Agreement Forest} (\AgreeForest)}
{Two unrooted binary phylogenetic trees $T_1, T_2$ on the same leaf set $X$, and a number $k$.}
{Partition $X$ into at most $k$ sets $S_1, \ldots, S_k$ where $T_1|S_i{=}T_2|S_i$ for each~$i$.}

Observation~\ref{obs:tiny} follows directly from the definitions. Observation \ref{obs:unbounded} shows that MAF and MRAF can behave very differently. 
\begin{observation}
\label{obs:tiny}
(a) A RAF with at most $\lceil \frac{n}{3} \rceil$ components always exists, where $n = |X|$, because if $|X'|=3$ and $X' \subseteq X$ we have $T_1|X' = T_2|X'$ irrespective of $X'$ or the topology of $T_1$ and $T_2$. (b) MRAF is 0 if and only if $T_1 = T_2$. (c) A partition $\{S_1, \ldots, S_k\}$ of $X$ is a RAF of $T_1, T_2$ if and only if, for each $S_i$, the set of quartets induced by $T_1|S_i$ is identical to the set of quartets induced by $T_2|S_i$.
\end{observation}
%
%MAF and MRAF can behave very differently:
%
\begin{observation}
\label{obs:unbounded}
There are instances where MAF is  arbitrarily large, $\Omega(n)$, while MRAF is constant.
\end{observation}
\begin{proof}
Let $T$ be an arbitrary unrooted phylogenetic binary tree on $n$ taxa. We create two trees $T_1$ and $T_2$, both on $4n$ taxa. We build $T_1$ by replacing each leaf $x$ in $T$ with a subtree on $\{a_x, b_x, c_x, d_x\}$ in which $a_x, b_x$ form a cherry and $c_x, d_x$ form a cherry. The construction of $T_2$ is similar except that $a_x, c_x$ form a cherry and $b_x, d_x$ form a cherry. Note that $T_1|\{a_x, b_x, c_x, d_x\} \neq T_2|\{a_x, b_x, c_x, d_x\}$. MRAF here is 2 because we can take one component containing all the $a_x, b_x$ taxa and one containing all the $c_x, d_x$ taxa. However, MAF is at least $n$. This is because in any AF
%for every $x$, due to the disjointness requirement,
at least one of the four taxa in $\{a_x, b_x, c_x, d_x\}$ must be a singleton component, %(i.e. a component with one taxon in it)
and there are $n$ subsets of the form $\{a_x, b_x, c_x, d_x\}$. \hfill $\qed$
\end{proof}

Given two trees $T_1, T_2$ on $X$ we say that $X' \subseteq X$ induces a \emph{maximum agreement subtree} (MAST) if $T_1|X' = T_2|X'$ and $X'$ has maximum cardinality ranging over all such subsets. Clearly, $\lceil \frac{n}{MAST}\rceil$ is a lower bound on MRAF, since each component of a RAF is no larger than a MAST. A MAST can be computed in polynomial time \cite{STEEL199377}. The trivial upper bound on MRAF of $\lceil \frac{n}{3} \rceil$ (see Observation \ref{obs:tiny}), which already contrasts sharply with the fact that the MAF of two trees can be as large as $n(1 - o(1))$ \cite{atkins2019extremal}, can easily be strengthened via MASTs. For example, it can be verified computationally or analytically that for any two trees on 6 or more taxa, a MAST has size at least 4. We can thus repeatedly choose and remove a homeomorphic size-4 subtree, until there are fewer than 6 taxa left, giving a loose upper bound on MRAF of $n/4 + 2$. In fact, it is known that the size of a MAST on two trees with $n$ leaves is $\Omega(\log n)$ \cite{markin2020extremal} (and that this bound is asymptotically tight). In particular, the lower bound on MAST grows to infinity as $n$ grows to infinity. Hence, the upper bound of $n/4 + 2$ can be strengthened to $n/c + f(c)$ for any arbitrary constant $c>1$ where $f$ is a function that only depends on $c$; this is thus $n/c + O(1)$. In fact, by iteratively removing $\Omega(\log n')$ of the \emph{remaining} number of taxa $n'$ we obtain a (slightly) sublinear upper bound on the size of a MRAF.
%$O(\frac{n}{\log \log n})$ is not too difficult to show.
Namely, while $n' \geq \log n + O(1)$, each iteration removes at least $d \log n' \geq d \log \log n$ leaves for some constant $d$, giving an upper bound of $\frac{n}{d\log \log n} + \log n + O(1)$ which is  $O(\frac{n}{\log \log n})$.

Regarding lower bounds, it is easy to generate pairs of
%unrooted binary phylogenetic
trees on $n$ leaves where a MAST has at most $O(\log n)$ leaves
%, e.g., when one tree is a caterpillar and the other is a complete binary tree 
\cite{kubicka1992agreement,markin2020extremal}.
%Each component of a RAF for these trees has size at most $O(\log n)$, so
A MRAF for such tree pairs will thus have size $\Omega(\frac{n}{\log n}$).

%\section{Partitioning Permutations and Co-Coloring}
\section{Hardness of MRAF}
\label{sec:np-hard}

We discuss a related NP-hard problem regarding partitioning permutations~\cite{Wagner84}.

\prob{\textsc{Partition into Monotone Subsequences (PIMS)}}
{A permutation $\pi$ of $\{1, \ldots, n\}$, and a number $k$.}
{Partition $\{1, \ldots, n\}$ into at most $k$ sets such that each set occurs monotonically in $\pi$, i.e., either as an increasing or a decreasing sequence.}

Due to the classical Erd{\H{o}}s~Szekeres Theorem~\cite{ErdosSzekeres1935}, for any $n$-element permutation there is a monotone sequence in $\pi$ with at least $\sqrt{n}$ elements. 
This can be used to efficiently partition $\pi$ into at most $2\sqrt{n}$ monotone sequences~\cite{Bar-YehudaF98}. 
Thus, we may assume that the $k$ in the problem statement is always at most $2\sqrt{n}$.  

%\section{NP-hardness}

%Here we show that \AgreeForest is NP-hard using the hardness of the PIMS problem. 

\begin{theorem}
MRAF is NP-hard.
\end{theorem}
\begin{proof}
Let $(\pi,k)$ be an input to the PIMS problem, i.e., $k$ is an integer greater than 1 and $\pi$ is a permutation of $\{1, \ldots, n\}$, where we use $\pi_i$ to denote the $i$th element of $\pi$. 
As remarked before, $k$ is at most $2\sqrt{n}$. 
This will imply that our constructed instance of \AgreeForest will have linear size in terms of the given permutation $\pi$, and as such any lower bounds, e.g., arising from the Exponential Time Hypothesis (ETH), will carry over from the PIMS problem to the \AgreeForest problem. 
For each pair of integers $(\alpha,\beta)$ where $\alpha+\beta=k$ and $\alpha, \beta \geq 1$\footnote{$\alpha=0$ or $\beta=0$ makes the problem easy.}, we will construct an instance $(T_1, T_2)$ of \AgreeForest such that $(T_1,T_2)$ has a solution consisting of $k$ trees if and only if $\pi$ can be partitioned into $\alpha$ increasing sequences and $\beta$ decreasing sequences. 
The trees $T_1$ and $T_2$ are described as follows. 

Recall that a \emph{caterpillar} is a tree $T$ where the subtree obtained by removing all leaves of $T$ is a path. 
The path here is called the \emph{spine} of the caterpillar.  Note that, in the caterpillars used to construct $T_1$ and $T_2$, some spine vertices will have degree~2. 
However, to make proper binary trees one should contract any such vertex into one of its neighbors.

We first construct a leaf set $v_1, \ldots, v_n$ corresponding to the permutation. 
We create an identity caterpillar $I$ whose spine is the $n$-vertex path $(x_1, \ldots, x_n)$ such that $x_i$ is adjacent to $v_i$. 
Next, we create a caterpillar $P$ whose spine is the $n$-vertex path $(y_1, \ldots, y_n)$ such that $y_i$ is adjacent to $v_{\pi_i}$. 
Observe that already for the \AgreeForest instance $(I,P)$, any $(r,s)$ partition of $\pi$ leads to a solution to $(I,P)$ consisting of $k$ trees. 
However, the converse is not yet enforced. 
In particular, if the input to MRAF is $(I,P)$,
then the components in a MRAF (which are  caterpillars) have cherries at their ends which, crucially, might be ordered differently in $I$ than in $P$. This can violate monotonicity.
\begin{figure}[bt]
%\hspace{-1.5cm}
        \includegraphics[scale=0.7]{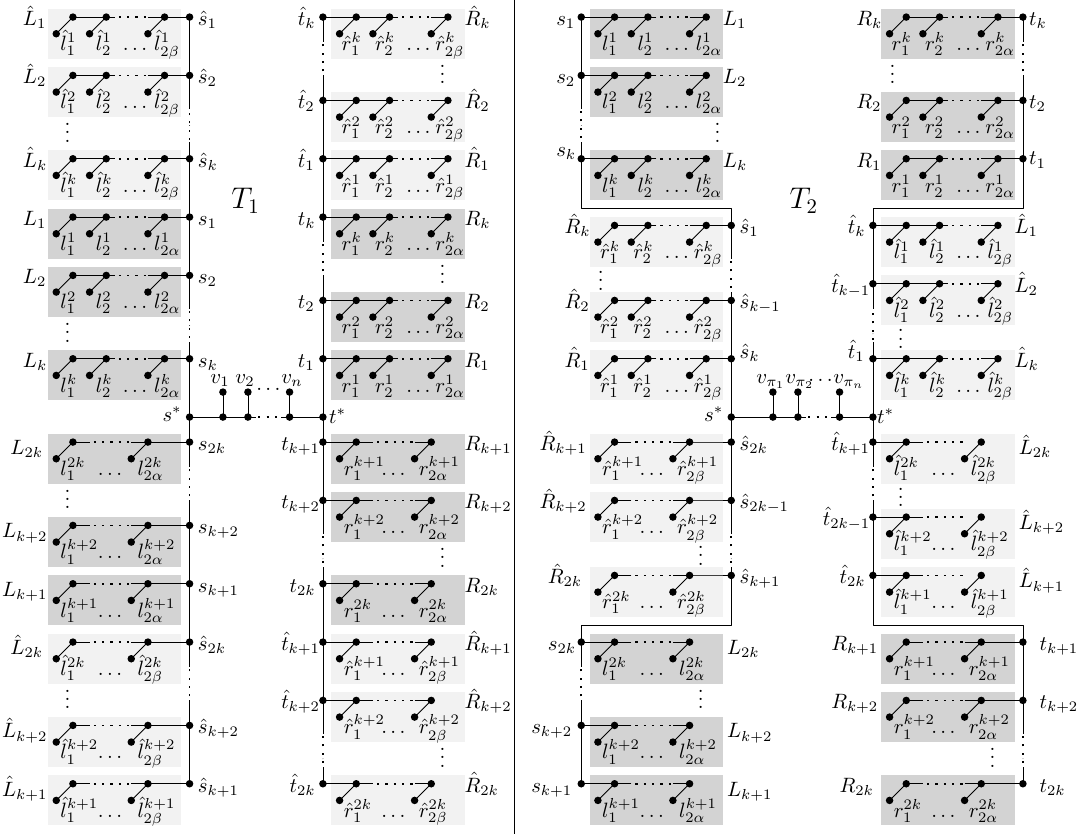}
    \caption{The two trees $T_1, T_2$ constructed from an instance of PIMS in the NP-hardness proof. The dark (light) grey leaves are used to induce increasing (decreasing) subsequences in the permutation-encoding taxa in the centre of the trees. }
    \label{fig:hardness}
\end{figure}
To counter this we extend 
$I$ and $P$ to obtain $T_1, T_2$ as shown in Fig.~\ref{fig:hardness}.  
For $T_1$, we construct $8k$ caterpillars. 
First, for the increasing sequences, we construct $4k$ caterpillars $L_1, \ldots, L_{2k},$ $R_1, \ldots, R_{2k}$ each having $2\alpha$ leaves and $2\alpha$ spine vertices.
Namely,~for~each~$i$, 
\begin{itemize}
    \item $L_i$ is the caterpillar with leaf set $\{l^i_1, \ldots, l^i_{2\alpha}\}$ and spine $(w^i_1, \ldots, w^i_{2\alpha})$ where, for each $j$, $l^i_j$ is adjacent to $w^i_j$; and
    \item $R_i$ is the caterpillar with leaf set $\{r^i_1, \ldots, r^i_{2\alpha}\}$ and spine $(z^i_1, \ldots, z^i_{2\alpha})$ where, for each $j$, $r^i_j$ is adjacent to $z^i_j$. 
\end{itemize}
For the decreasing sequences, we similarly construct $4k$ caterpillars $\hat{L}_1, \ldots, \hat{L}_{2k},$ $\hat{R}_1, \ldots, \hat{R}_{2k}$ each having $2\beta$ leaves and $2\beta$ spine vertices.
Namely, for each $i$, 
\begin{itemize}
    \item $\hat{L}_i$ is the caterpillar with leaf set 
    $\{\hat{l}^i_1, \ldots, \hat{l}^i_{2\beta}\}$ and spine $(\hat{w}^i_1, \ldots, \hat{w}^i_{2\beta})$ where, 
    for each $j$, $\hat{l}^i_j$ is adjacent to $\hat{w}^i_j$; and
    \item $\hat{R}_i$ is the caterpillar with leaf set $\{\hat{r}^i_1, \ldots, \hat{r}^i_{2\beta}\}$ and spine $(\hat{z}^i_1, \ldots, \hat{z}^i_{2\beta})$ where, for each $j$, $\hat{r}^i_j$ is adjacent to $\hat{z}^i_j$. 
\end{itemize}
To form $T_1$, we create two $(4k+1)$-paths $Q_{\mathrm{start}} = (\hat{s}_1, \ldots, \hat{s_k},$ $s_1, \ldots, s_{k},$ $s^*,$ 
$s_{2k}, \ldots, s_{k+1},$ 
$\hat{s}_{2k}, \ldots, \hat{s}_{k+1})$ and 
$Q_{\mathrm{end}} = (\hat{t}_k, \ldots, \hat{t}_1)$
$t_k, \ldots, t_1,$ $t^*,$ 
${t}_{k+1}, \ldots, \hat{t}_{2k}$
$\hat{t}_{k+1}, \ldots, \hat{t}_{2k})$ 
such that $s^*$ is adjacent to $x_1$ (i.e., to the ``start'' of $I$) and $t^*$ is adjacent to $x_n$ (i.e., to the ``end'' of $I$), and 
for each $i \in \{1, \ldots, 2k\}$:
\begin{itemize}
    \item $s_i$ is adjacent to $w^i_{2\alpha}$, i.e., the ``end'' of $L_i$ is attached to $s_i$, and $t_i$ is adjacent to $z^i_{1}$, i.e., the ``start'' of $R_i$ is attached to $t_i$; and 
    \item $\hat{s}_i$ is adjacent to $\hat{w}^i_{2\alpha}$, i.e., the ``end'' of $\hat{L}_i$ is attached to $\hat{s}_i$, and $\hat{t}_i$ is adjacent to $\hat{z}^i_{1}$, i.e., the ``start'' of $\hat{R}_i$ is attached to $\hat{t}_i$. 
\end{itemize}

To build $T_2$, we use the same $8k$ caterpillars $L_i, R_i, \hat{L}_i, \hat{R}_i$ but attach them differently to the ``central'' path $P$ of $T_2$.
First we make an adjustment to $Q_{\mathrm{start}}$ and $Q_{\mathrm{end}}$. In $T_2$, these become: 
$Q_{\mathrm{start}} = (s_1, \ldots, s_{k},$ 
$\hat{s}_1, \ldots, \hat{s_k},$ $s^*,$
$\hat{s}_{2k}, \ldots, \hat{s}_{k+1}$
$s_{2k}, \ldots, s_{k+1})$ and 
$Q_{\mathrm{end}} = (t_k, \ldots, t_1,$ 
$\hat{t}_k, \ldots, \hat{t}_1, $ $t^*,$ 
$\hat{t}_{k+1}, \ldots, \hat{t}_{2k},$ 
${t}_{k+1}, \ldots, \hat{t}_{2k})$ -- this swap is done to highlight that in $T_2$ the $\hat{L}_i, \hat{R}_i$ caterpillars are closer to the central path $P$ than the $L_i, R_i$ caterpillars. 
Similar to $T_1$, in $T_2$, we have $s^*$ adjacent to $y_1$ (i.e., the ``start'' of $P$) and $t^*$ is adjacent to $y_n$ (i.e., the ``end'' of $P$). 
The next part is where we see a difference regarding how we attach the caterpillars ($L_i, R_i$) of the increasing sequences vs. those ($\hat{L}_i, \hat{R}_i$) of decreasing sequences. 
For each $i \in \{1, \ldots, 2k\}$: 
\begin{itemize}
\item $s_i$ is adjacent to $w^i_{1}$, i.e., the ``start'' of $L_i$ is attached to $s_i$ and as such $L_i$ occurs ``reversed'' in $T_2$ with respect to $T_1$, and 
\item $t_i$ is adjacent to $z^i_{2\alpha}$, i.e., the ``end'' of $R_i$ is attached to $t_i$. 
\end{itemize}
For each $i \in \{1, \ldots, k\}$: 
\begin{itemize}
\item $\hat{s}_{k-i+1}$ ($\hat{s}_{2k-i+1}$) is adjacent to $\hat{z}^{i}_{2\beta}$ ($\hat{z}^{k+i}_{2\beta}$), i.e., the ``end'' of $\hat{R}_i$ ($\hat{R}_{i+k}$) is attached to $\hat{s}_{k-i+1}$ (and $\hat{s}_{2k-i+1}$) and as such $\hat{R}_i$ ($\hat{R}_{k+i}$) occurs ``on the opposite side'' in $T_2$ with respect to its location in $T_1$, and 
\item $\hat{t}_{k-i+1}$ ($\hat{t}_{2k-i+1}$) is adjacent to $\hat{w}^i_{1}$ ($\hat{w}^{k+i}_{1}$), i.e., the ``start'' of $\hat{L}_i$ ($\hat{L}_{k+i}$) is attached to $\hat{t}_{k-i+1}$ ($\hat{t}_{2k-i+i}$). 
\end{itemize}
This completes the construction of $T_1$ and $T_2$ from $\pi$. 
It is easy to see that this construction can be performed in polynomial time and that our trees contain precisely $16k^2 + 8k + 4 + 4n$ vertices, i.e., since $k \leq 2 \sqrt{n}$, our instance of \AgreeForest has $O(n)$ size. 

Suppose $\pi$ can be partitioned into $\alpha$ increasing sequences $\tau_1, \ldots, \tau_\alpha$ and $\beta$ decreasing sequences $\sigma_1, \ldots, \sigma_\beta$. 
The leaf set corresponding to $\tau_i$ consists of $\{v_{p} : p \in \tau_i\}$ together with two leaves from each of $L_j$ and $R_j$ $(j \in \{1, \ldots, 2k\})$, i.e., $l^j_{2i-1}, l^j_{2i}, r^j_{2i-1}, r^j_{2i-1}$.
Similarly, the leaf set corresponding to $\sigma_i$ consists of $\{v_{p} : p \in \sigma_i\}$ together with two leaves from each of $\hat{L}_j$ and $\hat{R}_j$ $(j \in \{1, \ldots, 2k\})$, i.e., $\hat{l}^j_{2i-1}, \hat{l}^j_{2i}, \hat{r}^j_{2i-1}, \hat{r}^j_{2i}$. It can be verified that this is a valid solution
to \AgreeForest.

Now suppose that we have a solution $S_1, \ldots, S_k$ to \AgreeForest$(T_1,T_2)$. 
We need to show that this leads to a solution to the PIMS problem on $\pi$ consisting of (at most) $\alpha$ increasing sequences and (at most) $\beta$ decreasing sequences. The proof of the following lemma is in the appendix.

\begin{lemma}\label{lem:threeleaves}
If some $S_j$ uses three leaves of any caterpillar $C \in \{L_i,R_i,\hat{L}_i,\hat{R}_i : i \in \{1, \ldots, 2k\}\}$ then all elements of $S_j$ are leaves of $C$. 
\end{lemma}

A consequence of this lemma is that if some $S_j$ uses more than two leaves from any single one of our left/right caterpillars, then $S_j$ can contain at most $\max\{2\alpha,2\beta\} < 2k$ elements. 
In the next part we will see that every $S_j$ must contain precisely $8k$ leaves from our left/right caterpillars in order to cover them all. 
In particular, this means that no $S_j$ contains more than two leaves from any single left/right caterpillar. 
Note that, the total number of leaves is 
$n + 4k \cdot 2\alpha + 4k \cdot 2\beta = n + 8k^2$ where the set of $n$ leaves is $\{v_1, \ldots, v_n\}$ (i.e., corresponding to the permutation) and the $8k^2$ leaves are the leaves of the left/right caterpillars. We now define the following eight leaf sets related to our caterpillars $L_i, R_i, \hat{L}_i, \hat{R}_i$. 
\begin{itemize}
\setlength{\itemsep}{0pt}
    \item $\mathcal{L}_1 = \{l : l $ is a leaf of some $L_i$, $i \in \{1, \ldots, k\}\}$, 
    \item $\mathcal{L}_2 = \{l : l $ is a leaf of some $L_i$, $i \in \{k+1, \ldots, 2k\}\}$, 
    \item $\mathcal{R}_1 = \{r : r $ is a leaf of some $R_i$, $i \in \{1, \ldots, k\}\}$, 
    \item $\mathcal{R}_2 = \{r : r $ is a leaf of some $R_i$, $i \in \{k+1, \ldots, 2k\}\}$. 
%    \item $\hat{\mathcal{L}}_1 = \{\hat{l} : \hat{l} $ is a leaf of some $\hat{L}_i$, $i \in \{1, \ldots, k\}\}$, 
%    \item $\hat{\mathcal{L}}_2 = \{\hat{l} : \hat{l} $ is a leaf of some $\hat{L}_i$, $i \in \{k+1, \ldots, 2k\}\}$, 
%    \item $\hat{\mathcal{R}}_1 = \{\hat{r} : \hat{r} $ is a leaf of some $\hat{R}_i$, $i \in \{1, \ldots, k\}\}$, 
%    \item $\hat{\mathcal{R}}_2 = \{\hat{r} : \hat{r} $ is a leaf of some $\hat{R}_i$, $i \in \{k+1, \ldots, 2k\}\}$. 
\end{itemize}
The definition of $\hat{\mathcal{L}}_1, \hat{\mathcal{L}}_2, \hat{\mathcal{R}}_1, \hat{\mathcal{R}}_2$ is analogous.
The proof of the following is also deferred to the appendix.
\begin{lemma}\label{lem:5sets}
No $S_j$ can contain five elements where each one belongs to a different set among: $\mathcal{L}_1, \mathcal{L}_2, \mathcal{R}_1, \mathcal{R}_2, \hat{\mathcal{L}}_1, \hat{\mathcal{L}}_2, \hat{\mathcal{R}}_1, \hat{\mathcal{R}}_2$. 
\end{lemma}

Now, observe that a component $S_j$ can contain
at most $2k$ taxa from each of the 8 sets listed above. That is because each of the 8 sets is formed from $k$ caterpillars (e.g., $\mathcal{L}_1$ is formed from the caterpillars $L_1, ..., L_k$)  and each of these $k$ caterpillars contributes at most 2 taxa to a RAF component. (If one of the $k$ caterpillars contributed more than 2 taxa, we would automatically be limited to at most $2k$ taxa, by Lemma \ref{lem:threeleaves}.) It follows from this that a component $S_j$ can in total intersect with at most $4 \times 2k = 8k$ taxa ranging over all the 8 sets: intersecting with more would require intersecting with at least 5 of the 8 sets, which as we have shown in Lemma \ref{lem:5sets} is not possible.

Given that there are $k$ components in the RAF, and $T_1, T_2$ have $n + 8k^2$ taxa, each of the $k$ components must
therefore contain \emph{exactly} $8k$ taxa from the 8 sets, and each component intersects
with \emph{exactly} 4 of the 8 sets (as this is the only way to achieve $8k$). In the appendix we prove that the only way for $S_j$ to intersect with four sets \emph{and} a permutation-encoding taxon $v_i$, is if the four sets are
$\{\mathcal{L}_1, \mathcal{L}_2, \mathcal{R}_1, \mathcal{R}_2 \}$ or $\{\hat{\mathcal{L}}_1, \hat{\mathcal{L}}_2, \hat{\mathcal{R}}_1, \hat{\mathcal{R}}_2 \}$. The permutation-encoding taxa $v_i$ contained in components
of the first type, necessarily induce increasing subsequences, and those contained in the second type are descending. There can be at most $\alpha$ components of the first type, and at most $\beta$ of the second, which means that the permutation $\pi$ can be partitioned into at most $\alpha$ increasing and $\beta$ decreasing sequences.
%(and thus at most $k$ subsequences in total).
This concludes the proof. \hfill $\qed$
\end{proof}

\section{Exact algorithms}

%\subsection{A single-exponential time exact algorithm for MRAF via Set Cover}

We now observe a single-exponential exact algorithm for MRAF and then show that when one input tree is a caterpillar, MRAF is in XP parameterized by $k$.

Recall that the NP-hard Set Cover problem $(U,F)$, where $F$ consists of subsets of $U$, is to compute a minimum-size subset of $F$ whose union is $U$. 
\begin{observation}
\label{obs:sc}
Let $T_1, T_2$ be two unrooted binary phylogenetic trees on $X$. Let $U=X$ and let $F$ be the set of all subsets of $X$ that induce homeomorphic trees in $T_1, T_2$. Each RAF of $T_1, T_2$ is a set cover of $(U,F)$, and each set cover of $(U,F)$ can be transformed in polynomial time into a RAF of $T_1, T_2$ with the same or smaller size, by allocating each element of $X$ to exactly one of the selected subsets. In particular, any optimum solution to the set cover instance $(U,F)$ can be transformed in polynomial time to yield a MRAF of $T_1, T_2$ of the same size. 
\end{observation}

\begin{lemma}
\label{lem:c_to_the_n}
MRAF can be solved in time $O(c^n)$, $n=|X|$, for some constant $c$.
\end{lemma}
\begin{proof}
The construction in Observation \ref{obs:sc} yields $|U|=n$ and $|F| \leq 2^n$. 
Minimum set cover can be solved in time $O( 2^{|U|} \cdot (|U|+|F|)^{O(1)})$
 thanks to~\cite{BjorklundHK09}. \hfill $\qed$
 \end{proof}
%\subsection{One caterpillar}

Lemma \ref{lem:c_to_the_n} concerns general instances.
When one of the given trees is a caterpillar, we can place \AgreeForest into XP (parameterized by the solution size $k$). 
We use dynamic programming for this.
We will assume that $n>3k$, as otherwise an arbitrary partition $S_1,\ldots,S_k$ where each $S_i$ has at most three taxa is a MRAF.
For $n>3k$ it follows that if there is a MRAF for $T_1$ and $T_2$, then there always is a MRAF $S_1,\ldots,S_k$ where no $S_i$ is a singleton. To see this, observe that for any MRAF with a singleton $S_i$, it must contain a component $S_j$ with $|S_j|\geq 3$ (since $n>3k$), and moving any element from $S_j$ to $S_i$ gives another MRAF where $S_i$ is not a singleton.

We let $T_1$ be the caterpillar, and $T_2$ an arbitrary tree. 
Similarly to our hardness result, we consider, without loss of generality, $T_1$ to consist of a spine (a path) $(y_1,\ldots,y_n)$ and leaves $v_1,\ldots,v_n$, where leaf $v_i$, $i=1,\ldots,n$ is adjacent to vertex $y_i$. See Fig.~\ref{fig:exact_algorithm_for_caterpillar} for an illustration.
The spine naturally orders the leaves (up to arbitrarily breaking ties on the end cherry taxa) and this will guide our dynamic-programming approach. 
We write $u\prec v$ for two leaves $u$ and $v$, if $u$ appears before $v$ in the considered ordering along the spine of $T_1$.
We decide whether a MRAF $S_1,\ldots,S_k$ of $T_1$ and $T_2$ exists as follows: we enumerate over all possible pairs of vertices $l_i, r_i$, $i=1,\ldots,k$, and check (compute) whether there exists a MRAF where the first leaf of $S_i$, $i=1,\ldots,k$, is $l_i$ and the last leaf of $S_i$ is $r_i$. We call such MRAF a \emph{MRAF constrained by $l_i,r_i$, $i=1,\ldots,k$}, or simply a \emph{constrained MRAF} if $l_i$ and $r_i$ are clear from the context.
If for one of the guesses (enumerations) we find a constrained MRAF, we output YES, and otherwise (if for all guesses we do not find a MRAF) we output NO.

\begin{figure}[t]
    \centering
    \includegraphics[scale=1]{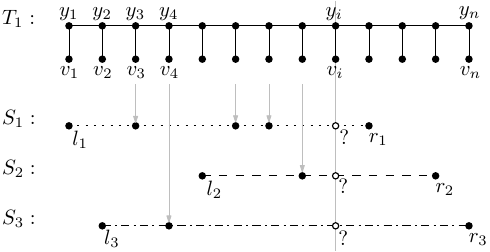}
    \caption{Caterpillar $T_1$ induces a natural ordering on the taxa (leaves). The gray vertical arrows assign each taxa to one of the sets $S_1$, $S_2$, $S_3$. At iteration $i$, the question marks denote possible assignment of $v_i$.}
    \label{fig:exact_algorithm_for_caterpillar}
\end{figure}

We now present our algorithm to decide, for input $T_1$, $T_2$, and pairs $l_i,r_i$, $i=1,\ldots, k$, whether a constrained MRAF exists. 
We define $L:=\{l_1,\ldots,l_n\}$ and $R:=\{r_1,\ldots,r_n\}$.
We view the process of computing constrained MRAF $S_1,\ldots, S_k$ as an iteration over $v_i$, $i=1,\ldots,n$, and assigning $v_i\notin (L\cup R)$ to one of the components $S_1,\ldots,S_k$ (every taxon $v_i\in (L\cup R)$ is already assigned). 
Fig.~\ref{fig:exact_algorithm_for_caterpillar} illustrates this by the gray arrows from each taxon to one of the sets $S_i$. 
In the constrained MRAF, taxon $v_i$ can only be assigned to component $S_j$ if and only if $l_j \prec v_i \prec r_j$.

Tree $T_2$ further limits how taxon $v_i$ can be assigned to components $S_j$ (because we want that $T_1|S_j=T_2|S_j$).
Clearly, for any $S_j\subset X$, $T_1|S_j$ is a caterpillar of maximum degree at most three. 
Thus, since $l_j$ and $r_j$ are the first and last leaf in $T_1|S_j$, they also need to be first and last in $T_2|S_j$. Hence, the inner vertices of the unique path $P_j$ from $l_j$ to $r_j$ in $T_2$ is the subdivision of the spine of $T_2|S_j$.
For a vertex $w\in P_j$ that has a neighbor $w'\notin P_j$ we define a \emph{bag $B_w$ of $P_2$} to be the maximal subtree of $T_2$ rooted at $w'$ that does not include $w$.
See Fig.~\ref{fig:spine_bags_in_T2} for an illustration.
Observe that for any bag $B_w$ of $P_j$, at most one taxon from $B_w$ can be assigned to $S_j$. 
(Because if two taxa $v_a,v_{a'}\in B_w$, $a<a'$, are assigned to $S_j$ then $l_j v_a | v_{a'} r_j$ will not be a quartet of $T_2|S_j$, while it is a quartet of $T_1|S_j$, and thus $T_1|S_j \neq T_2|S_j$.)
The path $P_j$ of $T_2|S_j$ naturally orders all bags of $P_j$. 
It follows that for two bags $B_w$ and $B_{w'}$ where $B_w$ appears before $B_{w'}$ in the ordering along $P_j$, we can select taxa $v_a\in B_w$ and $v_b\in B_{w'}$ into $S_j$ if and only if $v_a\prec v_b$, i.e., if $v_a$ appears before $v_b$ in the caterpillar $T_1$.
We write $v\prec_{P_j} v'$ for taxa $v$, $v'$ such that $v$ is from a bag $B_w$ and $v'$ is from a bag $B_{w'}$, and $B_w$ appears before bag $B_{w'}$ along path $P_j$.
Relation $\prec_{P_j}$ is thus a partial ordering of $X$, where any two taxa from the same bag are uncomparable.
Observe now that any assignment of taxa to $S_j$ that satisfies the above conditions, i.e.,  (i) for every $v_i\in S_j$, $l_j\prec v_i r_i$, (ii) for every bag $B_w$ of $P_j$ there is at most one vertex $v_i\in B_w\cap S_j$, and $(iii)$ for any two taxa $v_p,v_q\in S_j$, $p<q$, $v_p\prec_{P_j} v_q$, we have $T_1|S_j = T_2|S_j$.

\begin{figure}[t]
    \centering
    \includegraphics[scale=1]{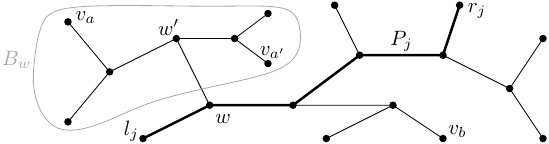}
    \caption{A bag $B_w$ on the $(l_j,r_j)$-path $P_j$. 
    At most one of $v_a,v_{a'}$ can occur in $S_j$.  %\textcolor{red}{Caption}
    }
    \label{fig:spine_bags_in_T2}
\end{figure}

We can thus assign taxon $v_i$ to component $S_j$ whenever the previously assigned taxon $v_{i'}$ to $S_j$ satisfies $v_{i'}\prec_{P_j} v_i$.
We thus do not need to know all previously assigned taxa to $S_j$, only the last assigned. 
We compute a (partial) restricted MRAF for taxa $X_i:=\{1,2,\ldots,i\}\cup E$ iteratively for $i=1,2,\ldots,k$. We set $X_0:=L\cup R$. For $\vec{z}=(z_1,\ldots,z_k)\in (X\setminus (L\cup R))^k$ and $i=0,1,\ldots,k$ we define a boolean function $\text{craf}^{(i)}(\vec{z})$ as follows: $\text{craf}^{(i)}(\vec{z}) := \text{TRUE}$ if and only if there exists a constrained MRAF
$S^{i}_1,S^{i}_2,\ldots,S^{i}_k$ of $X_i$ such that the last taxon from $X_i\setminus R$ in $S^{i}_\ell$, $\ell=1,\ldots,k$, is $z_\ell$.

Clearly, $\text{craf}^{(0)}(\vec{z})= \text{TRUE}$ if and only if $\vec{z}=(l_1,l_2,\ldots,l_k)$.
Also observe that if no $z_j$ is equal to taxon $v_i$, then $\text{craf}^{(i)}(z_1,\ldots,z_k)$ is FALSE, because in every partition of $X_i$, the last element $v_i$ of $X_i\setminus R$ needs to be last in one of the sets $S_j$.
Now, whenever one of $z_j$ is equal to $v_i$, the function $\text{craf}^{i}$ can be computed recursively as follows:

\begin{align}
   \label{eq:recursive_definition}
    \text{craf}^{(i)}(z_1,\ldots, z_{j-1},z_j & =v_i,z_{j+1} \ldots, z_k) = \\
    & \bigvee_{\substack{z\in X_{i-1}\setminus R \\ z \prec_{P_j} =v_i}}  \text{craf}^{(i-1)}(z_1,\ldots,z_{j-1},z,z_{j+1},\ldots,z_k) \nonumber
    %
    %\left[
    %\begin{matrix}
    %  z_1  \\
    %  \vdots  \\
    %  z_j=v_i \\
    %  \vdots \\
    %  z_k  \\
    %\end{matrix}
    %\right ]
\end{align}

This recurrence follows simply because removing $v_i$ from every constrained MRAF of $X_i$ gives a constrained MRAF of $X_{i-1}$.
Now we can compute $\text{craf}^{(i)}$ bottom-up using the dynamic programming. For every value $i=1,\ldots,k$ we enumerate $O(n^k)$ vectors $\vec{z}$, and compute the value $\text{craf}^{(i)}(\vec{z})$ using the recursive relation from Eq.~(\ref{eq:recursive_definition}), thus looking at at most $O(n)$ different entries of $\text{craf}^{(i-1)}$.
This thus leads to the overall runtime of $O(k\cdot n^k \cdot n$). 
Accounting for the enumeration of the $O(n^{2k})$ pairs $l_i,r_i$, $i=1,\ldots,k$ results in the following theorem.

\begin{theorem}
MRAF can be computed in time $O(k\cdot n^{3k+1})$ whenever one of the trees is a caterpillar.
\end{theorem}

\vspace{-0.5cm}
\section{Approximation algorithms}

We now provide a polytime approximation algorithm for MRAF  (Lemma~\ref{lem:logapprox}) and relate the approximability of PIMS to that of MRAF on caterpillars (Lemma~\ref{lem:pims-vs-mraf}). 

%\subsection{An approximation algorithm for general MRAF instances}

\begin{lemma}\label{lem:logapprox}
There is a polynomial-time $O(\log n)$ approximation for computing MRAF, where $n=|X|$. This cannot be better than a $(4/3)$ approximation.
\end{lemma}
\begin{proof}
Given an instance $(U,F)$ of Set Cover
%it is well known that
the natural greedy algorithm
%- repeatedly picking a subset from $F$ that covers a maximum number of uncovered elements of $U$ -
yields a $O(\log |U|)$ approximation. Recall the encoding of MRAF as a Set Cover instance in Observation \ref{obs:sc}. We cannot construct this
directly, since $|F|$ is potentially exponential in $n$, but this is not necessary to simulate the greedy algorithm. Let $X'$ be the set of currently uncovered elements of $X$, initially $X=X'$. We compute a MAST of $T_1|X'$ and $T_2|X'$ in polynomial time \cite{STEEL199377}. Let $S$ be the leaf-set of this MAST; we add this to our RAF. We then delete $S$ from $X'$ and iterate this process until $X'$ is empty. Fig.~\ref{fig:mastbad} (in the appendix) shows that %however one selects the next MAST in the above approximation algorithm,
this cannot be better than a $(4/3)$ approximation.
\hfill $\qed$
\end{proof}
%\noindent
%Note that two trees can potentially have many MASTs.
%\vspace{-.5cm}

%\smallskip

%\noindent\textbf{Approximating PIMS vs. approximating MRAF on caterpillars}
%the approximability of PIMS}
%\smallskip

%We can relate the approximability of \AgreeForest when both trees are caterpillars to that of the PIMS problem. 
%For a given instance of PIMS there is an instance of \AgreeForest (consisting of two caterpillars) such that any size $k$ solution to the \AgreeForest problem leads to a size $k+2\sqrt{2k}$ solution to the PIMS problem. 
%This is formalized in the following lemma. 

\begin{lemma}\label{lem:pims-vs-mraf}
Let $\pi$ be a permutation of $\{1, \ldots, n\}$ and let $T_1$ and $T_2$ be two caterpillars on leaves $\{1, \ldots n\}$ where $T_1$ is the identity caterpillar and  the $i$th leaf of $T_2$ is $\pi(i)$.
For any solution to the MRAF problem of size $k$, there is a corresponding solution to the
PIMS problem
of size at most $k + 2\sqrt{2k}$. 
\end{lemma}

\begin{proof}
 
 %   Clearly, any partition of $\pi$ into monotone subsequences leads to a valid agreement forest for the caterpillars. Thus, an optimal agreement forest cannot contain more elements than the monotone partition.  
    
    %In the other direction, w
    We start with an agreement forest for the two caterpillars; each component is itself a caterpillar. We ``cut off'' one leaf from each end of the components in this forest. (This is because the ``interior'' of each component induces a monotonic subsequence, but the cherries at the end of each component potentially violate this). 
    This leaves behind a subpermutation of $\pi$ of length $2k$, which can always be partitioned into at most $2\sqrt{2k}$ monotone subsequences. \hfill $\qed$
\end{proof}
We can create an instance of PIMS from a caterpillar instance of \AgreeForest by treating one caterpillar as the identity and the other as the permutation. 
Any solution for this PIMS instance yields a feasible MRAF solution. Hence:
\[
\text{MRAF} \leq \text{PIMS} \leq \text{MRAF} + 2\sqrt{2 \cdot \text{MRAF}}.
\]
Recall that MRAF on caterpillars is in XP, Theorem~\ref{fig:exact_algorithm_for_caterpillar}.
%Now, we have already proven that MRAF on caterpillars is in XP. 
PIMS is also in XP. Specifically, the PIMS problem is equivalent to the \emph{co-chromatic number} problem on permutation graphs, i.e., partitioning the vertices of a permutation graph into
%the fewest sets so that each set is either a
cliques and independent sets.  When a graph can be partitioned into $r$ cliques and $s$ independent sets it is sometimes called an \emph{$(r,s)$-split graph}. It is known that the perfect $(r,s)$-split graphs can be characterized by a finite set of forbidden induced subgraphs~\cite{KezdySW96}. This
implies that their recognition is in XP parameterized by $r$ and $s$, i.e., when $r$ and $s$ are fixed, $(r,s)$-split graphs can be recognized in polynomial time---this was later improved to FPT~\cite{HeggernesKLRS13}. These XP results are relevant here because they mean that if one of the problems has a polynomial time $c$-approximation, $c$ constant, then for each fixed constant $\epsilon > 0$ the other has a polynomial-time $(c+\epsilon)$-approximation. For example, given a polynomial-time $c$-approximation for MRAF, and $\epsilon > 0$, we first use the XP algorithm for PIMS to check in polynomial time whether $\text{PIMS} \leq \frac{8c}{\epsilon^2}$. If so, we are done. Otherwise the described transformation of MRAF solutions to PIMS solutions yields a $(c+\epsilon)$-approximation. The direction from PIMS to MRAF is similar. PIMS has a polynomial-time 1.71-approximation \cite{fomin2002approximating}. Hence, for every constant $\epsilon > 0$ MRAF on caterpillars has a polynomial-time $(1.71+\epsilon)$-approximation.
\section{Implementation and experimental observations}
%Towards computing MRAF in practice,
MRAF can be modelled as the \emph{weak chromatic number} of hypergraph: the minimum number of colours assigned to vertices, such that no hyperedge is monochromatic. The set of vertices is $X$ and there is a hyperedge $\{a,b,c,d\}$ whenever the two trees have a different quartet topology on $\{a,b,c,d\}$. This leverages the characterization implied by Observation \ref{obs:tiny}.
%\begin{hidden}
%Colours represent the component of the forest a leaf is in; a component cannot contain all of $\{a,b,c,d\}$ simultaneously (cf. the hyperedge cannot be coloured monochromatically) because that would mean the component induces different quartet topologies on $\{a,b,c,d\}$ in the two input trees. There are $O(|X|^4)$ hyperedges.
%\end{hidden}
We implemented this as a constraint program (CP) using the MiniZinc solver \cite{nethercote2007minizinc}. For trees with up to around 30 leaves the CP solves quickly.
%We use this later in the article to verify a number of mathematical claims and to empirically assess how MRAF compares to other phylogenetic parameters/distances in an applied setting.
The code is available at 
\url{ https://github.com/skelk2001/relaxed_agreement_forests}.
We used this to
%analyse the gap between MAF and MRAF in a more practical setting,
extend the analysis of \cite{kelk2016monadic} on the grass dataset of \cite{grass2001}, consisting of fifteen pairs of trees. See Table \ref{tab:mraf}; as expected MRAF grows rather more slowly than MAF. 
In fact, MRAF seems more comparable to the \emph{treewidth} of the \emph{display graph} of the input tree pair. (The display graph of $T_1, T_2$ is obtained by identifying vertices with the same leaf label: the treewidth of this graph is bounded by a function of MAF \cite{kelk2016monadic}.) We obtained similar results on a larger and more challenging dataset, comprising the 163 tree pairs from the %synthetic
dataset in \cite{van2022reflections} that had at most 50 leaves after pre-processing.
%The dataset is more challenging because many of the tree pairs are highly dissimilar;
See Table \ref{tab:summary} in the appendix.

%The fact that this function grows to infinity as a function of $n$ means that an analytical upper bound, based on interatively removing $\Omega(\sqrt{\log n' })$ of the remaining number of taxa $n'$,\\

%\vspace{-0.75cm}
\section{Discussion and open problems}
%We list here a number of open problems.
It remains unclear whether it is NP-hard to compute MRAF on caterpillars, although it seems likely.
Can the finite forbidden obstructions that characterize solutions to PIMS somehow be mapped to MRAF on caterpillars? Could this then be generalized to MRAF on general trees? Indeed, how far can MRAF be viewed as a generalization of the PIMS problem to partial orders? Is MRAF on caterpillars FPT? Does it (or PIMS) have a polynomial kernel? What should reduction rules look like, given that rules for MAF seem to be of  limited use (see Appendix \ref{appendix:reduc})? Strikingly, we do not know whether it is NP-hard to determine whether MRAF $\leq 2$ for two 
general trees,
%unrooted binary trees $T_1, T_2$,
meaning that the FPT landscape is also unclear.
%(with MRAF as the parameter).
%The  
How far can the logarithmic approximation factor for MRAF on general trees, and the 1.71 approximation for MRAF on caterpillars (equivalently, PIMS) be improved? Finally, it would be instructive to elucidate the potential biological interpretation of this model.
%
%
%
%
%\vspace{-1cm}

\begin{table}[t]
%\scriptsize
\small
\centering
\begin{tabular}{|l|c|c|c|c|c|c|}
\hline
\textbf{tree pair}        & \multicolumn{1}{l|}{$|X|=n$} & \multicolumn{1}{l|}{\textbf{MAF}} & \multicolumn{1}{l|}{\textbf{MRAF}} & \multicolumn{1}{l|}{\textbf{tw(D)}} & \multicolumn{1}{l|}{\textbf{MAST}} & \multicolumn{1}{l|}{\textbf{$\lceil n/MAST \rceil$}} \\ \hline
\textit{00\_rpoC\_waxy.txt} & 10                                     & 2                                 & \textbf{2}                         & 3                                   & 8                                  & 2                                          \\ \hline
\textit{01\_phyB\_waxy.txt} & 14                                     & 3                                 & \textbf{2}                         & 3                                   & 11                                 & 2                                          \\ \hline
\textit{02\_phyB\_rbcL.txt} & 21                                     & 5                                 & \textbf{3}                         & 3                                   & 14                                 & 2                                          \\ \hline
\textit{03\_rbcL\_waxy.txt} & 12                                     & 4                                 & \textbf{2}                         & 3                                   & 9                                  & 2                                          \\ \hline
\textit{04\_phyB\_rpoC.txt} & 21                                     & 5                                 & \textbf{2}                         & 3                                   & 15                                 & 2                                          \\ \hline
\textit{05\_waxy\_ITS.txt}  & 15                                     & 6                                 & \textbf{3}                         & 4                                   & 10                                 & 2                                          \\ \hline
\textit{06\_phyB\_ITS.txt}  & 30                                     & 8                                 & \textbf{4}                         & 4                                   & 17                                 & 2                                          \\ \hline
\textit{07\_ndhF\_waxy.txt} & 19                                     & 5                                 & \textbf{3}                         & 4                                   & 11                                 & 2                                          \\ \hline
\textit{08\_ndhF\_rpoC.txt} & 34                                     & 9                                 & \textbf{3}                         & 5                                   & 20                                 & 2                                          \\ \hline
\textit{09\_rbcL\_rpoC.txt} & 26                                     & 7                                 & \textbf{4}                         & 5                                   & 14                                 & 2                                          \\ \hline
\textit{10\_ndhF\_rbcL.txt} & 36                                     & 7                                 & \textbf{4}                         & 3                                   & 20                                 & 2                                          \\ \hline
\textit{11\_rbcL\_ITS.txt}  & 29                                     & 11                                & \textbf{4}                         & 5                                   & 17                                 & 2                                          \\ \hline
\textit{12\_ndhF\_phyB.txt} & 40                                     & 7                                 & \textbf{3}                         & 3                                   & 30                                 & 2                                          \\ \hline
\textit{13\_rpoC\_ITS.txt}  & 31                                     & 11                                & \textbf{4}                         & 6                                   & 16                                 & 2                                          \\ \hline
\textit{14\_ndhF\_ITS.txt}  & 46                                     & 16                                & \textbf{5}                         & 6                                   & 20                                 & 3                                          \\ \hline
\end{tabular}
\smallskip
\caption{Comparison of MAF and MRAF for the fifteen tree pairs in the data set \cite{grass2001} analysed in \cite{kelk2016monadic}. We also include MAST,
%(the maximum number of leaves in a homeomorphic subtree of the tree pair)
the lower bound on MRAF given by $\lceil \frac{n}{MAST} \rceil$, and $tw(D)$ which is the treewidth of the display graph obtained from the tree pair.}
\label{tab:mraf}
\end{table}
%\vspace{-1cm}
\bibliographystyle{plain}
\bibliography{sample}

\begin{thebibliography}{10}

\bibitem{AllenSteel2001}
B.~Allen and M.~Steel.
\newblock Subtree transfer operations and their induced metrics on evolutionary
  trees.
\newblock {\em Annals of Combinatorics}, 5:1--15, 2001.

\bibitem{atkins2019extremal}
R.~Atkins and C.~McDiarmid.
\newblock Extremal distances for subtree transfer operations in binary trees.
\newblock {\em Annals of Combinatorics}, 23:1--26, 2019.

\bibitem{Bar-YehudaF98}
R.~Bar{-}Yehuda and S.~Fogel.
\newblock Partitioning a sequence into few monotone subsequences.
\newblock {\em Acta Informatica}, 35(5):421--440, 1998.

\bibitem{BjorklundHK09}
A.~Bj{\"{o}}rklund, T.~Husfeldt, and M.~Koivisto.
\newblock Set partitioning via inclusion-exclusion.
\newblock {\em {SIAM} Journal on Computing}, 39(2):546--563, 2009.

\bibitem{bordewich2005computational}
M.~Bordewich and C.~Semple.
\newblock On the computational complexity of the rooted subtree prune and
  regraft distance.
\newblock {\em Annals of Combinatorics}, 8(4):409--423, 2005.

\bibitem{bulteau2019parameterized}
L.~Bulteau and M.~Weller.
\newblock Parameterized algorithms in bioinformatics: an overview.
\newblock {\em Algorithms}, 12(12):256, 2019.

\bibitem{chen2016approximating}
J.~Chen, F.~Shi, and J.~Wang.
\newblock Approximating maximum agreement forest on multiple binary trees.
\newblock {\em Algorithmica}, 76:867--889, 2016.

\bibitem{degnan2009gene}
J.~Degnan and N.~Rosenberg.
\newblock Gene tree discordance, phylogenetic inference and the multispecies
  coalescent.
\newblock {\em Trends in Ecology \& Evolution}, 24(6):332--340, 2009.

\bibitem{ErdosSzekeres1935}
P.~Erd{\H{o}}s and G.~Szekeres.
\newblock A combinatorial problem in geometry.
\newblock {\em Compositio Mathematica}, 2:463--470, 1935.

\bibitem{fomin2002approximating}
F.~Fomin, D.~Kratsch, and J-C. Novelli.
\newblock Approximating minimum cocolorings.
\newblock {\em Information Processing Letters}, 84(5):285--290, 2002.

\bibitem{grass2001}
Grass Phylogeny~Working Group and N.~Barker et~al.
\newblock Phylogeny and subfamilial classification of the grasses (poaceae).
\newblock {\em Annals of the Missouri Botanical Garden}, 88(3):373--457, 2001.

\bibitem{HeggernesKLRS13}
P.~Heggernes, D.~Kratsch, D.~Lokshtanov, V.~Raman, and S.~Saurabh.
\newblock Fixed-parameter algorithms for cochromatic number and disjoint
  rectangle stabbing via iterative localization.
\newblock {\em Information and Computation}, 231:109--116, 2013.

\bibitem{HeinJWZ96}
J.~Hein, T.~Jiang, L.~Wang, and K.~Zhang.
\newblock On the complexity of comparing evolutionary trees.
\newblock {\em Discrete Applied Mathematics}, 71(1-3):153--169, 1996.

\bibitem{kelk2016monadic}
S.~Kelk, L.~van Iersel, C.~Scornavacca, and M.~Weller.
\newblock Phylogenetic incongruence through the lens of monadic second order
  logic.
\newblock {\em Journal of Graph Algorithms and Applications}, 20(2):189--215,
  2016.

\bibitem{KezdySW96}
A.~K{\'{e}}zdy, H.~Snevily, and C.~Wang.
\newblock Partitioning permutations into increasing and decreasing
  subsequences.
\newblock {\em Journal of Combinatorial Theory, Series {A}}, 73(2):353--359,
  1996.

\bibitem{kubicka1992agreement}
E.~Kubicka, G.~Kubicki, and F.~McMorris.
\newblock On agreement subtrees of two binary trees.
\newblock {\em Congressus Numerantium}, pages 217--217, 1992.

\bibitem{lemey2009phylogenetic}
P.~Lemey, M.~Salemi, and A-M. Vandamme.
\newblock {\em The phylogenetic handbook: a practical approach to phylogenetic
  analysis and hypothesis testing}.
\newblock Camrbdige U.P., 2009.

\bibitem{markin2020extremal}
A.~Markin.
\newblock On the extremal maximum agreement subtree problem.
\newblock {\em Discrete Applied Mathematics}, 285:612--620, 2020.

\bibitem{nakhleh2013computational}
L.~Nakhleh.
\newblock Computational approaches to species phylogeny inference and gene tree
  reconciliation.
\newblock {\em Trends in Ecology \& Evolution}, 28(12):719--728, 2013.

\bibitem{nethercote2007minizinc}
N.~Nethercote, P.~Stuckey, R.~Becket, S.~Brand, G.~Duck, and G.~Tack.
\newblock {M}ini{Z}inc: Towards a standard {C}{P} modelling language.
\newblock In {\em CP2007}, pages 529--543, 2007.

\bibitem{steel2016phylogeny}
M.~Steel.
\newblock {\em Phylogeny: Discrete and Random Processes in Evolution}.
\newblock SIAM, 2016.

\bibitem{STEEL199377}
M.~Steel and T.~Warnow.
\newblock Kaikoura tree theorems: Computing the maximum agreement subtree.
\newblock {\em Information Processing Letters}, 48(2):77--82, 1993.

\bibitem{van2022reflections}
R.~van Wersch, S.~Kelk, S.~Linz, and G.~Stamoulis.
\newblock Reflections on kernelizing and computing unrooted agreement forests.
\newblock {\em Annals of Operations Research}, 309(1):425--451, 2022.

\bibitem{Wagner84}
K.~Wagner.
\newblock Monotonic coverings of finite sets.
\newblock {\em Journal of Information Processing and Cybernetics},
  20(12):633--639, 1984.

\end{thebibliography}

\clearpage
\appendix

\section{Appendix}

\subsection{Omitted proofs}

\textbf{Lemma }\ref{lem:threeleaves}
\emph{If some $S_j$ uses three leaves of any caterpillar $C \in \{L_i,R_i,\hat{L}_i,\hat{R}_i : i \in \{1, \ldots, 2k\}\}$ then all elements of $S_j$ are leaves of $C$.} 
\begin{proof}
    Let $C \in \{L_i,R_i,\hat{L}_i,\hat{R}_i : i \in \{1, \ldots, 2k\}\}$ and suppose $S_j$ contains three leaves $a,b,c$ of $C$, and one leaf $d$ that is not in $C$. Without loss of generality, let $a, b, c$ be ordered in increasing distance from the permutation-encoding part of $T_1$ (i.e., $I$); there may be a tie between $b$ and $c$ but this does not matter. Observe that $T_1$ induces the quartet $da|bc$ but $T_2$ induces the quartet $cd|ba$ or $bd|ca$. This is because in $T_2$, $C$ is attached to the rest of the tree by the opposite end used to attach $C$ to the rest of $T_1$.
    Recall that a necessary and sufficient condition for $S_j$ to be a component of a RAF is that they induce exactly the same set of quartet topologies in both trees (Observation \ref{obs:tiny}(c)); contradiction. \hfill $\qed$
%    Without loss of generality, suppose the the neighbor of $b$ occurs between the neighbor $a$ and the neighbor of $c$ on the spine of $C$. 
\end{proof}
\noindent
\textbf{Lemma \ref{lem:5sets}.}
\emph{No $S_j$ can contain five elements where each one belongs to a different set among: $\mathcal{L}_1, \mathcal{L}_2, \mathcal{R}_1, \mathcal{R}_2, \hat{\mathcal{L}}_1, \hat{\mathcal{L}}_2, \hat{\mathcal{R}}_1, \hat{\mathcal{R}}_2$.} 
\begin{proof}
Let $a,b,c,d,e \in S_j$ be chosen from 5 distinct sets from the 8 listed.
Consider the four set pairs
$\{\hat{\mathcal{L}}_1, {\mathcal{L}}_1\}$,
$\{{\mathcal{L}}_2,\hat{\mathcal{L}}_2\}$,
$\{\hat{\mathcal{R}}_1, {\mathcal{R}}_1\}$,
$\{{\mathcal{R}}_2,\hat{\mathcal{R}}_2\}$ which together partition the 8 sets. Note that the two sets in each pair are ``adjacent'' in $T_1$, but in $T_2$ they are on opposite sides of the permutation-encoding part $P$. By the pigeonhole principle at least one of these four set pairs must have elements from $a,b,c,d,e$ in both sets of the pair. Without loss of generality, suppose $a \in \hat{\mathcal{L}}_1$ and
$b \in \mathcal{L}_1$. Then in $T_1|\{a,b,c,d,e\}$, leaves $\{a,b\}$ form a cherry. However, due to $a$ and $b$ being on opposite sides of $P$ in $T_2$, their options for forming a cherry there are highly constrained. If, say, $c \in \mathcal{R}_1$ then $|S_j|\leq 3$ because $\{a,c\}$ then forms a cherry in $T_2|\{a,b,c,d,e\}$ and the only way for a taxon (here $a$) to be in two distinct cherries in $T_1|S_j = T_2|S_j$, is if $S_j$ has exactly 3 leaves. The same analysis holds if $c \in \hat{\mathcal{R}_1}$. Hence, again by the pigeonhole principle, at least one taxon from $\{c,d,e\}$ must be in $\hat{\mathcal{R}}_2 \cup \mathcal{L}_2$, and at least one taxon from $\{c,d,e\}$ must be in $\hat{\mathcal{L}}_2 \cup \mathcal{R}_2$. But then $\{a,b\}$ is certainly not a cherry in $T_2|\{a,b,c,d,e\}$. Hence, $T_1|\{a,b,c,d,e\} \neq T_2|\{a,b,c,d,e\}$; contradiction. \hfill $\qed$
\end{proof}

\noindent
\emph{Proof that the only way for $S_j$ to intersect with 4 sets \emph{and} a permutation-encoding taxon $v_i$, is if the 4 sets are}
$\{\mathcal{L}_1, \mathcal{L}_2, \mathcal{R}_1, \mathcal{R}_2 \}$ \emph{or} $\{\hat{\mathcal{L}}_1, \hat{\mathcal{L}}_2, \hat{\mathcal{R}}_1, \hat{\mathcal{R}}_2 \}$.\footnote{The GitHub page for this article includes an alternative and independent computational verification of this fact, based on enumerating all MASTs of two 9-taxon trees using constraint programming.}
\begin{proof}
Recall the four pairs $\{\hat{\mathcal{L}}_1, {\mathcal{L}}_1\}$,
$\{{\mathcal{L}}_2,\hat{\mathcal{L}}_2\}$,
$\{\hat{\mathcal{R}}_1, {\mathcal{R}}_1\}$,
$\{{\mathcal{R}}_2,\hat{\mathcal{R}}_2\}$ which are ``adjacent'' in $T_1$. To this list we
can add four other pairs, which are the sets which are
``adjacent'' in $T_2$. These are  
$\{{\mathcal{L}}_1,\hat{\mathcal{R}}_1\}$,
$\{\hat{\mathcal{R}}_2, {\mathcal{L}}_2\}$,
$\{{\mathcal{R}}_1,\hat{\mathcal{L}}_1\}$,
$\{\hat{\mathcal{L}}_2, {\mathcal{R}}_2\}$. All these in total 8 pairs have
the property that they are ``adjacent'' in one of the the two input trees,
but split across the permutation-encoding part of the other. Suppose $S_j$ contains taxa from both sets in one of these 8 pairs, $\{\hat{\mathcal{L}}_1, \mathcal{L}_1\}$ say (the other cases are symmetrical). Here $S_j$ has only two ways to intersect with two further sets whilst ensuring the same topology in both trees: (1) $\{\hat{\mathcal{R}}_2, {\mathcal{L}}_2\}$, (2) $\{\hat{\mathcal{L}}_2, {\mathcal{R}}_2\}$. However, whether (1) is chosen or (2), if $S_j$ contains some permutation-encoding taxon $v_i$, $v_i$ will be in a different location (with respect to these four sets) in $T_1$ than in $T_2$, contradicting that $T_1|S_j = T_2|S_j$. This means that for each of the 8 pairs, $S_j$ must
avoid intersecting with both the sets in the pair. This leaves at most $4 \times 4 = 16$ possible valid combinations: one of the 4 pairs
$\{\mathcal{L}_1,\mathcal{L}_2\}$,
$\{\mathcal{L}_1,\hat{\mathcal{L}_2}\}$,
$\{\hat{\mathcal{L}_1},\mathcal{L}_2\}$,
$\{\hat{\mathcal{L}_1},\hat{\mathcal{L}_2\}}$ from the left side of $T_1$,
and one of the 4 pairs
$\{\mathcal{R}_1,\mathcal{R}_2\}$,
$\{\mathcal{R}_1,\hat{\mathcal{R}_2}\}$,
$\{\hat{\mathcal{R}_1},\mathcal{R}_2\}$,
$\{\hat{\mathcal{R}_1},\hat{\mathcal{R}_2\}}$ from the right side of $T_1$. With some checking it can be verified that, of these 16, only $\{\mathcal{L}_1, \mathcal{L}_2, \mathcal{R}_1, \mathcal{R}_2 \}$ and $\{\hat{\mathcal{L}}_1, \hat{\mathcal{L}}_2, \hat{\mathcal{R}}_1, \hat{\mathcal{R}}_2 \}$ induce the same quartet topology in $T_1, T_2$. \hfill $\qed$
\end{proof}

\subsection{Reduction rules}
\label{appendix:reduc}
We observe that if two trees $T_1, T_2$ have a common cherry $\{a,b\}$, the well-known \emph{common cherry reduction} - in each tree, we delete $a, b$ and relabel their parent $ab$ - preserves MRAF. When applied to exhaustion this is called the \emph{subtree} reduction rule. This is known to be very effective in the phylogenetics literature when pre-processing input trees to reduce their size -- and will thus help with (exact) computation of MRAF in practice, given its NP-hardness. On the other hand, the much-studied \emph{common chain} reduction rule is not safe. The definition of this reduction rule is rather technical (see e.g. \cite{van2022reflections} for a formal definition) but essentially it shrinks two long common caterpillars to two shorter caterpillars. See Fig. \ref{fig:chain}: as can be verified with our code, shortening the common chain lowers MRAF. This is in contrast to MAF, where both the subtree and common chain reduction rules preserve MAF, and in fact yield a linear kernel \cite{AllenSteel2001}. We note that if a pair of trees has no common cherries or common chains, the ratio $\frac{n}{MRAF}$ can still be arbitrarily large. For example, two caterpillars of the form $1, z, 2, y, 3, x...$ and $1, a, 2, b, 3, c...$ have no common cherries or chains, and a MRAF of size 2, but an arbitrarily large number of leaves. Hence, any attempt to establish a fixed parameter tractability result for MRAF via kernelization must consider a different strategy.
\begin{figure}[h]
\centering
\includegraphics[scale=0.15]{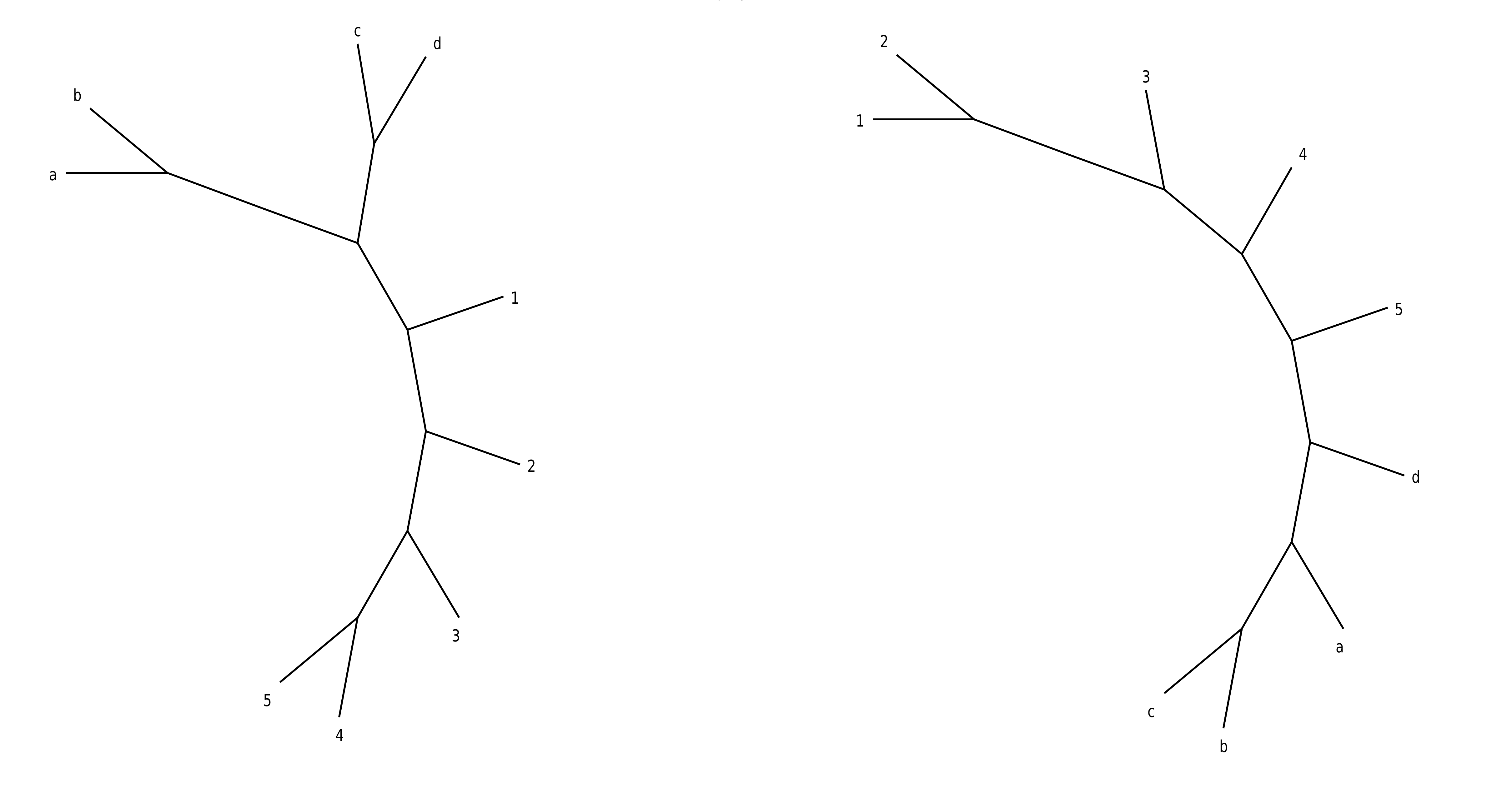}
\caption{Two trees with a common chain of size 5, comprising the leaves
$C=\{1,2,3,4,5\}$. The MRAF of these trees is 3. However, deleting any one taxon from $C$ yields a pair of trees with MRAF equal to 2.}
\label{fig:chain} 
\end{figure}
%\clearpage

%\subsection{Lower bound on the quality of the greedy approximation algorithm}

\begin{figure}[h]
\centering
\includegraphics[scale=0.15]{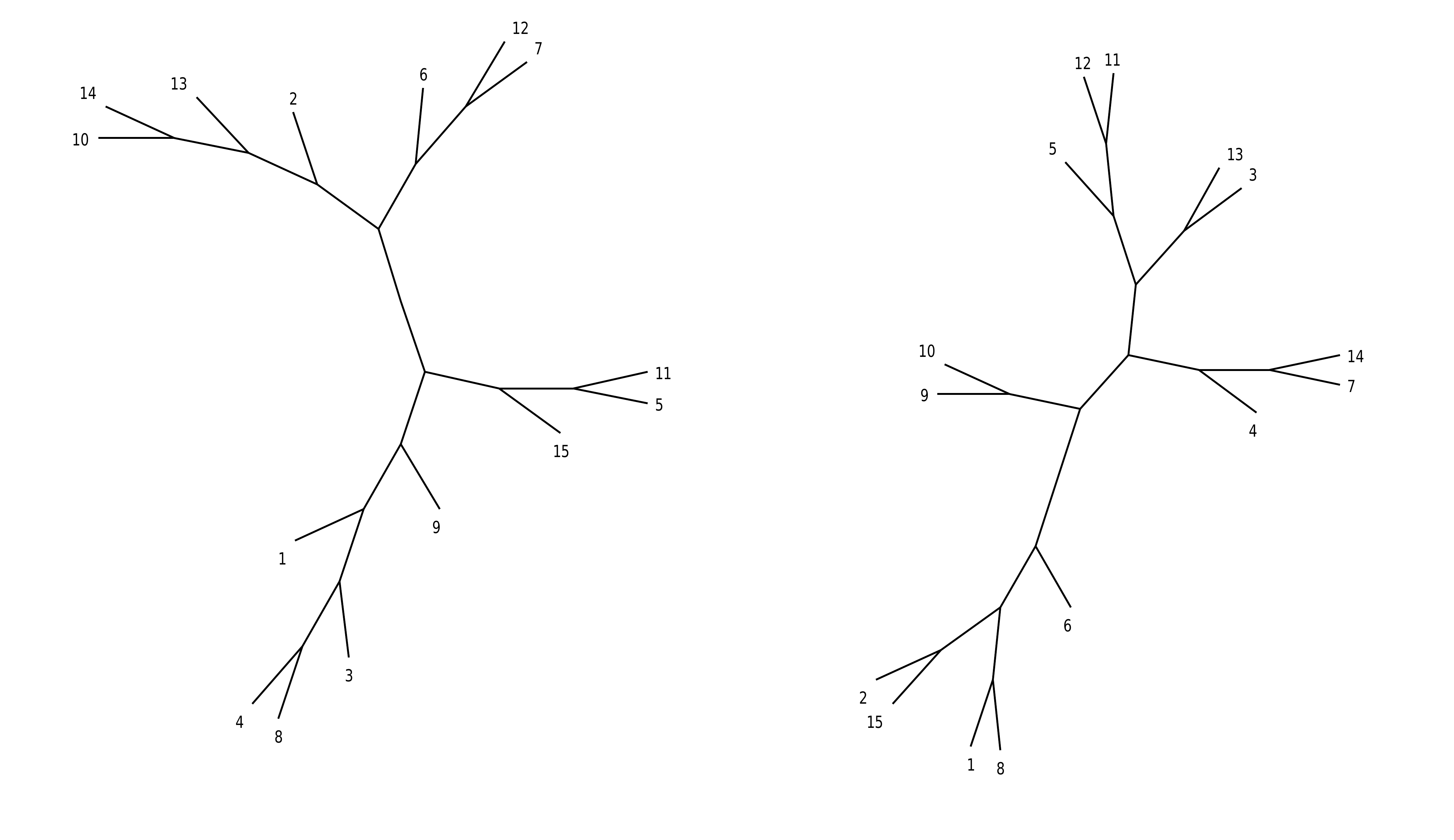}
\caption{These two trees have a MRAF of size 3, $\{1,7,8,10,12,15\}, \{5,11,13,14\}, \{2,3,4,6,9\}\}$ and a unique MAST of size 7: $\{1,5,7,8,9,11,14\}$. However, any RAF that contains a component of size 7, has 4 or more components. This has been verified computationally. Due to the uniqueness of the MAST, this shows that even an algorithm that can non-deterministically decide what the `correct' MAST is to choose next, cannot always solve MRAF optimally. We also have a very similar example, also proving a $(4/3)$ lower bound, where both trees are caterpillars. This example, and the aforementioned computational proof, are available on the GitHub page.}
\label{fig:mastbad} 
\end{figure}

%\subsection{Summary of results for the second dataset}

\begin{table}[h]
\centering
\scriptsize
\begin{tabular}{r|c|c|c|c|}
\cline{2-5}
\multicolumn{1}{l|}{}                   & \textbf{Min} & \textbf{Max} & \textbf{Avg} & \textbf{Stdev} \\ \hline
\multicolumn{1}{|r|}{\textit{MAF}}      & 5            & 27           & 13.07        & 7.19           \\ \hline
\multicolumn{1}{|r|}{\textit{MRAF}}     & 2            & 8            & 4.47         & 1.27           \\ \hline
\multicolumn{1}{|r|}{\textit{MRAF-MAF}} & -20          & -1            & -8.60        & 6.14           \\ \hline
\multicolumn{1}{|r|}{\textit{tw}}       & 3            & 13           & 7.28         & 2.46           \\ \hline
\multicolumn{1}{|r|}{\textit{MRAF-tw}}  & -1           & 7            & 2.81         & 1.58           \\ \hline
\end{tabular}
\caption{Summary statistics for the 163 tree pairs obtained from the dataset in \cite{van2022reflections} by restricting to trees which, after subtree reduction, have at most 50 taxa.}
\label{tab:summary}
\end{table}

\end{document}